\documentclass{amsart}
\usepackage{amssymb,amscd,amsthm}
\usepackage{latexsym}
\usepackage[final]{graphicx}
\newcommand{\Z}{{\mathbb Z}}
\newcommand{\R}{{\mathbb R}}
\newcommand{\C}{{\mathbb C}}
\newcommand{\N}{{\mathbb N}}

\newtheorem{theorem}{Theorem}
\newtheorem{remark}{Remark}[section]
\newtheorem{lemma}[remark]{Lemma}
\newtheorem{proposition}[remark]{Proposition}

\begin{document}

\title[Simplicity of eigenvalues]{Simplicity of eigenvalues in Anderson-type models}

\author{Sergey Naboko}

\address{Department of Mathematical Physics, Institute of Physics, St.\ Petersburg State University, St.\ Petersburg, 198504, Russia}

\email{naboko@math.su.se}

\thanks{S.\ N.\ was supported by Russian research grant RFBR  09-01-00515a.}

\author{Roger Nichols}

\address{Department of Mathematics, University of Alabama at Birmingham, Birmingham, AL~35294, USA}

\email{rnich02@uab.edu}

\author{G\"unter Stolz}

\address{Department of Mathematics, University of Alabama at Birmingham, Birmingham, AL~35294, USA}

\email{stolz@math.uab.edu}

\thanks{G.\ S.\ was supported in part by NSF grant DMS-0653374.}

\begin{abstract}
We show almost sure simplicity of eigenvalues for several models of Anderson-type random Schr\"odinger operators, extending methods introduced by Simon for the discrete Anderson model. These methods work throughout the spectrum and are not restricted to the localization regime. We establish general criteria for the simplicity of eigenvalues which can be interpreted as separately excluding the absence of local and global symmetries, respectively. The criteria are applied to Anderson models with matrix-valued potential as well as with single-site potentials supported on a finite box.

\end{abstract}

\maketitle

\section{Introduction}

\subsection{Models}

Some time back Barry Simon published the short note \cite{Simon} in which he proved almost sure simplicity of eigenvalues of the discrete Anderson model. The latter is the random operator acting on $u\in \ell^2(\Z^d)$ as
\begin{equation} \label{disAnderson}
(h_{\omega} u)(n) = (h_0 u)(n) + \omega_n u(n),
\end{equation}
where $h_0$ is the discrete Laplacian,
\begin{equation} \label{disLaplace}
(h_0u)(n) = \sum_{k\in\Z^d: |k-n| =1} u(k),
\end{equation}
and $\omega = (\omega_n)_{n\in \Z^d}$ are i.i.d.\ real random variables with distribution $\mu$. Here we assume that $\mu$ is absolutely continuous with bounded and compactly supported density $\rho$. While stated somewhat differently in \cite{Simon}, the result proven there can most easily be formulated as
\begin{theorem} \label{thm:simon}
For almost every $\omega$, all eigenvalues of $h_{\omega}$ are simple.
\end{theorem}

What makes this result particularly appealing is that it is known that the Anderson model has intervals of dense pure point spectrum. In fact, for sufficiently large disorder (in the sense that $\|\rho\|_{\infty}$ is sufficiently small) it is known that the entire spectrum of $h_{\omega}$ is almost surely pure point, e.g.\ \cite{Carmona/Lacroix}. Theorem~\ref{thm:simon} says that on intervals of pure point spectrum the spectral multiplicity of $h_{\omega}$ is one.

One of the reasons for being interested in results like Theorem~\ref{thm:simon} is that they can be useful tools in proofs of other properties of random operators, see e.g.\ the proof of dynamical localization in \cite{dRJLS}. However, our interest in Simon's result and the technique used to prove it comes mainly from the fact that it makes rigorous sense out of the following physical heuristics:

Degeneracies of eigenvalues, with the exception of accidental ones, are caused by symmetry. Randomness breaks all symmetry and accidental degeneracies should have probability zero. Thus an operator which is truly random should have simple eigenvalues with probability one.

One of the difficulties in making such heuristics rigorous lies in the fact that the connection between symmetry and eigenvalue degeneracy is usually understood via analytic perturbation theory: Analytic eigenvalue branches will either show permanent degeneracies (reflecting a symmetry not broken by the perturbation) or have level crossings only for discrete sets of the perturbation parameter. However, analytic perturbation theory does not apply to dense lying eigenvalues!

Another problem comes with the relative vagueness of the claim that randomness breaks symmetry. Do our favorite models of random operators come with the ``true randomness'' which rules out all symmetries, even potentially well-hidden ones?

It is mostly for these reasons that we have decided to give Simon's result a second, closer, look. We do this by considering three different models, where attempting to extend Simon's result causes an increasing amount of difficulty and technical complexity, while all of them fall under the same physical heuristics.

The first model contains the discrete Anderson model (\ref{disAnderson}) as a special case and will serve as a simple test case for the methods to be developed.

\subsection*{Model A: Anderson model with matrix-valued potential}

Fix $k\in \N$ and a real-valued, symmetric and positive definite $k\times k$-matrix $W$. Consider the random operator $H_{\omega}^A$ acting on $\phi \in \ell^2(\Z^d;\C^k) \cong (\ell^2(\Z^d))^k$ as
\begin{equation} \label{matAnderson}
(H_{\omega}^A \phi)(n) = (h_0 \phi)(n) + \omega_n W\phi(n), \quad n\in \Z^d.
\end{equation}
Here $(\omega_n)$ and $h_0$ are as above (more precisely, $h_0$ acts on each component of $\phi$ by (\ref{disLaplace})).

Without loss of generality we may assume that $W = \mbox{diag}(\lambda_1,\ldots,\lambda_k)$, where the eigenvalues $\lambda_j$ of $W$ are all strictly positive (apply the diagonalizing transformation of $W$ to each $n$ in (\ref{matAnderson}), leading to a unitarily equivalent operator). As a result,
\[ H_{\omega}^A \cong \bigoplus_{j=1}^k h_{\omega}^{(j)}, \]
where $(h_{\omega}^{(j)}u)(n) = (h_0u)(n) + \lambda_j \omega_n u(n)$ for $u\in \ell^2(\Z^d)$. Thus each $h_{\omega}^{(j)}$ is of the form (\ref{disAnderson}) with the additional parameter $\lambda_j$ scaling the random potential. Note, however, that the random operators $h_{\omega}^{(j)}$, $j=1,\ldots,k$, are correlated and that simplicity of the eigenvalues of $H_{\omega}^A$ is not an immediate consequence of Theorem~\ref{thm:simon}. In fact, if $W$ has degenerate eigenvalues, then the point spectrum of $H_{\omega}^A$ will have degeneracies of at least the same multiplicity with probability one. Thus we will need to require simplicity of $W$.

A more complex generalization of the Anderson model (\ref{disAnderson}) is given by

\subsection*{Model B: Anderson model with finitely supported single-site potential}

Choose $L=(L_1,L_2,\ldots,L_d)\in \mathbb{N}^d$
and consider the rectangular box
\[
C_0:=\{0,\ldots,L_1-1\}\times \cdots \times \{0,\ldots,L_d-1\}.
\]
Pick a {\it single-site potential} $f:C_0 \rightarrow (0,\infty)$ and for $n=(n_1,n_2\ldots,n_d) \in \Z^d$ let $nL:=(n_1L_1,n_2L_2,\ldots,n_dL_d)$.  Model B is the family of
self-adjoint operators $H_{\omega}^B$ on
$\ell^2(\Z^d)$ given by
\begin{equation} \label{modelB}
  H_{\omega}^B=h_0+\sum_{n\in \Z^d}\omega_n f(\cdot-nL).
\end{equation}

Models A and B have in common that they give generalizations of the discrete Anderson model where the single-site potential is an operator of finite rank greater than one, providing internal structure to the single-site terms. On the heuristic level of symmetry considerations one is lead to distinguish between global and local symmetries. The global symmetries are the ones which are broken by the randomness of the potential. But degeneracies within the single-site terms $W$ and $f$, respectively, give rise to additional symmetries, whose influence on the multiplicity of eigenvalues in the Anderson model is not clear, not even heuristically. From this point of view, particularly interesting special cases of Model B are those were $f= \chi_{C_0}$, the characteristic function of the box $C_0$. Here the single-site contributions to the random potential have maximal degeneracy and it must depend on the specifics of the interaction of potential and kinetic energy if these degeneracies can be broken up by the randomness.

A particular reason for introducing Model B is that it can be seen as a hybrid which shares some properties with the discrete Anderson model (\ref{disAnderson}) but has other features in common with our last model, the continuum Anderson model.

\subsection*{Model C: Continuum Anderson model}

This is the random operator in $L^2(\R^d)$ given by
\begin{equation} \label{contAnderson}
H_{\omega}^C = -\Delta + \sum_{n\in \Z^d} \omega_n f(x-n).
\end{equation}
Here $\Delta$ is the continuum Laplacian and the random parameters $(\omega_n)$ are as before. The single site potential is now a multiplication operator by a non-negative bounded function $f$, supported on $[0,1]^d$.

\subsection{Results}

While we hope to return to the continuum Anderson model in the future, we do not have any final results on the simplicity of its point spectrum to present here. Our concrete results on simplicity of the point spectrum will be restricted to Models A and B. However, in Section~\ref{sec:abstract} below we will start by presenting Theorem~\ref{thm:criterion}, a general criterion for simplicity of eigenvalues in terms of simplicity of corresponding Birman-Schwinger operators, which applies to all the models considered here. The criterion will yield two conditions which need to be verified in concrete examples to conclude simplicity. Physically, these conditions can be interpreted as absence of local and global symmetries, respectively. As discussed in detail at the end of Section~\ref{sec:abstract}, it is illuminating to see how the goal of verifying these two conditions brings out the mathematical differences between Models A, B and C. For Model A both conditions are relatively easy to verify, which makes it a nice test case. In Model B each condition yields additional challenges and, at least for one of the conditions, our answer will require additional information on the structure of $C_0$ and/or $f$. Finally, for the continuum Anderson model one condition is obviously true (it follows from unique continuation which is not available for the discrete models) while the other condition is very hard to check with any degree of generality (and we will not try here).

After the general results in Section~\ref{sec:abstract}, the rest of the paper is devoted to cases where the single-site potential is a finite rank operator, i.e.\ in particular to Models A and B. In Section~\ref{sec:simon} we present Theorem~\ref{thm:simon2}, a result which extends a rank-one argument provided in \cite{Simon} to a finite rank setting suitable for our applications.

This will be used in Section~\ref{sec:modelA} to prove simplicity of the point spectrum for Model A:

\begin{theorem} \label{thm:modelA}
Suppose that the positive definite matrix $W$ in (\ref{matAnderson}) has simple eigenvalues. Then $H_{\omega}^A$ has simple point spectrum for almost every $\omega$.
\end{theorem}

As noted above the simplicity of $W$ is necessary here.

More effort will go in the subsequent investigation of Model B, where our main result will be

\begin{theorem} \label{thm:modelB}
Suppose that for Model B one of the following additional assumptions holds:

(i) $d$ and $C_0$ are arbitrary and $f:C_0 \to (0,\infty)$ simple, i.e.\ $f(j) \not= f(k)$ for arbitrary $j, k\in C_0$ with $j\not= k$, or

(ii) $d$ arbitrary, $C_0 = \{0,\ldots,L_1-1\} \times \{0\} \times \ldots \times \{0\}$, $L_1$ any positive integer, $f= \chi_{C_0}$, the characteristic function of $C_0$, or

(iii) $d=2$, $C_0 = \{0,1\} \times \{0,1\}$, $f=\chi_{C_0}$.

Then the point spectrum of $H_{\omega}^B$ is almost surely simple.
\end{theorem}

Theorem~\ref{thm:modelB} will be proven in the last two sections, with Section~\ref{sec:weakcyclicModB} establishing the condition for absence of global symmetries and Section~\ref{sec:BSsimpleModB} showing the absence of local symmetries. Here the additional assumptions (i), (ii) or (iii) required in Theorem~\ref{thm:modelB} reflect different mechanisms which can be used to break local symmetries. In case (i) local symmetries are broken by the potential energy term alone. For (ii), where the single-site potentials are essentially one-dimensional, we will be able to use that one-dimensional Jacobi matrices have simple eigenvalues. Condition (iii) is the hardest but also most interesting case. Here we will have to use properties of the random environment (i.e.\ the effect of random variables other than $\omega_0$) to break the symmetries. While we can not prove simplicity of the point spectrum for Model B in full generality, the study of case (iii) provides prototypes of some techniques which would have to be pushed further (and understood in a way which uses less brute force) for a general result.

\subsection{Context}

An interesting alternative approach to simplicity of eigenvalues in the Anderson model, using methods very different from those employed here, has been found by Klein and Molchanov \cite{KM}. Their methods work in the localization regime, i.e.\ in energy regions where the spectrum is known to be pure point. They exploit known decay properties of Green's function in these regions together with the Minami estimate. The latter can be interpreted as showing the stochastic independence of near lying eigenvalues and served as the central tool in the proof of Poisson level statistics of the eigenvalues of finite volume restrictions of the discrete Anderson model in \cite{Minami}.

The proof of the Minami estimate and Poisson statistics has recently been extended to the continuum Anderson model by Combes, Germinet and Klein \cite{CGK}, where it holds in the localized region near the bottom of the spectrum. Their work also extends the result of \cite{KM} to the continuum Anderson model, showing almost sure simplicity of eigenvalues in the energy regime covered by \cite{CGK}.

Much of our motivation for the current investigation came from these works. While we can not treat the continuum Anderson model at this point, our methods establish simplicity of eigenvalues throughout the spectrum and are not restricted to the localized regime. One of our hopes is that we can use them in the future to show that a Minami estimate holds throughout the spectrum for general classes of Anderson-type models, ultimately including the continuum Anderson model, as shown for the discrete Anderson model in \cite{Minami}. We also refer to \cite{GV} and \cite{BHS} for other proofs of the Minami estimate as well as extensions to $n$-level Minami estimates. Far reaching extensions of the results in \cite{Minami} and \cite{CGK} were recently announced by Germinet and Klopp as work in preparation.

Finally, we mention work by Jaksic and Last \cite{JL} which extends Simon's result Theorem~\ref{thm:simon} to showing that the singular spectrum (the unions of point and singular continuous spectrum) of the discrete Anderson model is almost surely simple. We guess that this will also hold in more general situations like those considered here, but have not proven this. Jaksic and Last mention that results of this form allow for an intriguing way of viewing the extended states conjecture (or at least a weak version of it): If one could identify regimes with spectral regions of multiplicity larger than one in the Anderson model, then this would necessarily imply the existence of continuous spectrum (or, using their result, absolutely continuous spectrum).

\section{Simplicity through Birman-Schwinger operators} \label{sec:abstract}

The set
of eigenvalues of a selfadjoint operator $A$ on a separable Hilbert space $\mathcal H$ will be denoted by
$\sigma_p(H)$. For Borel sets $B\subset \R$ we denote by $\chi_B(A)$
the spectral projection onto $B$ for $A$. The closed linear span of all eigenfunctions of $A$ will be denoted by ${\mathcal H}^{pp}(A)$, the pure point subspace for $A$, and $P^{pp}(A)$ is the orthogonal projection onto ${\mathcal H}^{pp}(A)$.

If $M\subset {\mathcal H}$, then ${\mathcal H}(A,M)$ denotes the smallest reducing subspace for $A$ containing $M$. Note the characterization
\begin{equation} \label{eq:redsub}
 {\mathcal H}(A,M) = \overline{\mbox{span}\{ (A-z)^{-1}f: \,z\in \C\setminus \R, \, f\in M\}}.
\end{equation}

We will denote Lebesgue measure on $\R$ by $|\cdot|$. We use $N(\cdot)$ to denote null spaces and $R(\cdot)$ to denote ranges.

Let $H_0$ be a selfadjoint operator and $V$ a non-negative and
bounded operator in $\mathcal H$. Consider
the family of selfadjoint operators
\begin{equation} \label{eq:oneparmodel}
 H_{\lambda} = H_0 +\lambda V
 \end{equation}
for $\lambda \in \R$. We write ${\mathcal H}_V := {\mathcal H}(H_0,R(V))$ and note that it is easily seen from the resolvent identity and (\ref{eq:redsub}) that
\begin{equation} \label{eq:lambdaind}
{\mathcal H}_V = {\mathcal H}(H_{\lambda},R(V)) \quad \mbox{for all $\lambda \in \R$}.
\end{equation}

We will use the following consequence of spectral averaging:

\begin{lemma} \label{specaverage}
If $M\subset \R$ is such that $|M|=0$, then, for Lebesgue almost every $\lambda$, $\chi_M(H_{\lambda})|_{{\mathcal H}_V} =0$. In particular, for a.e.\ $\lambda$, $H_{\lambda}|_{{\mathcal H}_V}$ has no eigenvalues in $M$.
\end{lemma}

\begin{proof}
By spectral averaging, see Corollary~4.2 in \cite{Combes/Hislop} and its proof, it holds for arbitrary $\varphi \in {\mathcal H}$ and arbitrary Borel sets $B$ that
\[
\int_{\R} \frac{\langle \chi_B(H_{\lambda}) \sqrt{V} \varphi, \sqrt{V} \varphi \rangle}{1+\lambda^2} \,d\lambda \le |B| \|\varphi\|^2.
\]
Thus, if $|M|=0$ and $\varphi$ is fixed, then
\[
\langle \chi_B(H_{\lambda}) \sqrt{V} \varphi, \sqrt{V} \varphi \rangle = 0 \quad \mbox{for a.e.\ $\lambda$}.
\]
As ${\mathcal H}$ is separable, this implies
\[
\sqrt{V} \chi_B(H_{\lambda}) \sqrt{V} = 0 \quad \mbox{for a.e.\ $\lambda$}.
\]
For each such $\lambda$ it follows that $\chi_B(H_{\lambda})|_{{\mathcal H}_V} =0$: Observe first that for $f=V\phi \in R(V)$ and $\psi := \sqrt{V}\phi$,
\[
\|\chi_B(H_{\lambda})f\|^2 = \langle \sqrt{V} \chi_B(H_{\lambda}) \sqrt{V} \psi,\psi \rangle = 0,
\]
i.e.\ $\chi_B(H_{\lambda})f=0$. That $\chi_B(H_{\lambda})f =0$ for all $f\in {\mathcal H}_V$ follows easily from this and (\ref{eq:redsub}) (with $H_{\lambda}$ in place of $H_0$).
\end{proof}

Below we will consider the Birman-Schwinger operators
\[
G(z):= \sqrt{V} (H_0-z)^{-1} \sqrt{V}
\]
for $z\in \C\setminus \R$, as well as their operator-norm boundary values
\begin{equation} \label{eq:BSboundval}
G(E+i0) := \lim_{\varepsilon\downarrow 0} G(E+i\varepsilon),
\end{equation}
for $E\in \R$ where this boundary value exists.

The following abstract criterion will be the basis of all our further investigations.

\begin{theorem} \label{thm:criterion}
Assume that $G(z)$ is compact for all $z\in \C\setminus \R$ and that its boundary value $G(E+i0)$ exists for Lebesgue-almost every $E\in \R$ and has simple non-zero eigenvalues. Then
$H_{\lambda}|_{{\mathcal H}_V}$ has simple point spectrum for almost
every $\lambda \in\R$.
\end{theorem}

Before proving this, several comments are in order:

(i) Note that $G(E+i0)$ is compact if it exists. Thus its non-zero
spectrum consists entirely of discrete eigenvalues. By simplicity we
mean algebraic simplicity, i.e.\ all generalized eigenspaces are
one-dimensional (and thus eigenspaces). Note that $G(E+i0)$ is not necessarily selfadjoint.

(ii) The following proof shows that
there is also a local version of the result: If $I\subset \R$ is an
open interval and $G(E+i0)$
has simple non-zero eigenvalues for almost every $E\in I$, then, for
almost every $\lambda \in\R$, $H_{\lambda}|_{{\mathcal H}_V}$ has
simple point spectrum in $I$.

(iii) Theorem~\ref{thm:criterion} does not say anything
about the continuous spectrum of $H_{\lambda}$. When applying
Theorem~\ref{thm:criterion} to Anderson-type models, simplicity of
the entire spectrum of $H_{\lambda}$ (in some interval or the whole
line) follows if spectral localization, i.e.\ absence of continuous
spectrum, is established by separate means.

\begin{proof}
By assumption there exists a set $S\subset \R$ with $|S|=0$ such that,
for every $E\in \R\setminus S$, $G(E+i0)$ exists and has simple
non-zero eigenvalues. Let $M:= S \cup \sigma_p(H_0)$. By Lemma~\ref{specaverage} there
exists a set $A\subset \R$ with $|A|=0$ such that, for every
$\lambda \in \R\setminus A$, $H_{\lambda}|_{{\mathcal H}_V}$ has no
eigenvalues in $M$.

Fix $\lambda \in \R\setminus (A\cup \{0\})$. We will show that all
eigenvalues of $H_{\lambda}|_{{\mathcal H}_V}$ are simple. As
$|A\cup \{0\}|=0$, this proves the Theorem.

Let $E$ be an eigenvalue of $H_{\lambda}|_{{\mathcal H}_V}$. Thus
$E\not\in M$ and, in particular, $\chi_{\{E\}}(H_0)=0$. Also, $G(E+i0)$
exists and has simple non-zero eigenvalues. We will show that the
operator $\sqrt{V}$ defines a one-to-one mapping from
$N((H_{\lambda}-E)|_{{\mathcal H}_V})$ into
$N(G(E+i0)+\frac{1}{\lambda})$.

Let $u\in N((H_{\lambda}-E)|_{{\mathcal H}_V})$ and $u\not= 0$. Then
$H_{\lambda}u=Eu$, which is equivalent to
\begin{equation} \label{eq:equiv}
u= -\lambda (H_0-E-i\varepsilon)^{-1} \sqrt{V} \sqrt{V} u
-i\varepsilon (H_0-E-i\varepsilon)^{-1} u
\end{equation}
for every $\varepsilon>0$. Note that
\[ -i\varepsilon (H_0-E-i\varepsilon)^{-1}
\stackrel{s}{\longrightarrow} \chi_{\{E\}}(H_0) = 0 \]
as $\varepsilon \downarrow 0$. Multiplying \eqref{eq:equiv} by $\sqrt{V}$ and taking
$\varepsilon \downarrow 0$, we infer
\[ \sqrt{V} u = -\lambda G(E+i0) \sqrt{V} u. \]
This shows that $\sqrt{V}u \in N(G(E+i0)+\frac{1}{\lambda})$. Also,
$\sqrt{V}u \not= 0$ as otherwise it would follow from
$H_{\lambda}u=Eu$ that $H_0u=Eu$, a contradiction to $E\not\in M$.
Thus the mapping is one-to-one. From
\[ \mbox{dim} N((H_{\lambda}-E)|_{{\mathcal H}_V})\, \le \,\mbox{dim}
N(G(E+i0)+\frac{1}{\lambda}) \le 1 \]
we conclude that $E$ is a
simple eigenvalue of $H_{\lambda}|_{{\mathcal H}_V}$.
\end{proof}

We devote the rest of this section to a preliminary discussion of how one can hope to apply Theorem~\ref{thm:criterion} to prove simplicity of the point spectrum for Models A, B and C. First, we introduce language which allows to discuss the three models simultaneously.

Thus let $H_{\omega}$ be one of the operators $H_{\omega}^A$, $H_{\omega}^B$ or $H_{\omega}^C$.  $V_j$ is the action of the single-site potential at site $j$, i.e.
\[
(V_j\phi)(n) = \left\{ \begin{array}{ll} W\phi(j), & n=j, \\ 0, & n\not= j, \end{array} \right.
\]
for $\phi \in \ell^2(\Z^d;\C^k)$ in case of Model A, $(V_j\phi)(n) = f(n-jL)\phi(n)$, $n\in \Z^d$, for Model B, and $(V_j\phi)(x) =f(x-j)\phi(x)$, $x\in \R^d$ for Model C. For all three models we can now write
\[ H_{\omega} = H_0 + \sum_{j\in \Z^d} \omega_j V_j,\]
where $H_0$ is either the discrete or continuum Laplacian.

If we also denote by $P_j$ the orthogonal projection onto $R(V_j)$, then we have at least for Models A and B that $\sum_{j\in \Z^d} P_j = I$, a ``covering condition''. This is guaranteed by our assumptions, since $W>$ gives for Model A that $R(V_j) = \{\phi \in \ell^2(\Z^d;\C^k):\,\phi(n)=0 \;\mbox{for all}\;n\not= j\}$ and $f>0$ for Model B means that $R(V_j) = \{\phi\in \ell^2(\Z^d): \phi(n) =0\; \mbox{for all}\;n\not\in C_j\} = \ell^2(C_j)$. Here the {\it tiles} $C_j := C_0 -jL$ are the supports of $f(\cdot-jL)$, $j\in \Z^d$.

The coupling constant $\lambda$ in (\ref{eq:oneparmodel}) is identified with one of the random parameters which we choose to be $\omega_0$.
Writing $V=V_0$ and $\omega = (\hat{\omega}, \omega_0)$, where $\hat{\omega} = (\omega_n)_{n\not= 0}$, we have for all three models
\[
H_{\omega} = H_{\hat{\omega}} + \omega_0 V,
\]
which for fixed $\hat{\omega}$ takes the form of (\ref{eq:oneparmodel}). While we will often keep $\hat{\omega}$ fixed and study the effect of adding $\omega_0 V$ to $H_{\hat{\omega}}$, we stress that we can only expect to prove simplicity of the eigenvalues of $H_{\omega}$ for almost every $\omega$, i.e.\ almost every $\hat{\omega}$ and almost every $\omega_0$. Properties of the ``random environment'' $\hat{\omega}$ will play a role.

Our goal is to show that almost surely $H_{\omega}$ has simple point spectrum, i.e.\ that $H_{\omega}|_{{\mathcal H}^{pp}(H_{\omega})}$ has simple spectrum. Obviously, this follows if we can establish the following two properties:

\begin{equation} \label{eq:BSsimple}
H_{\omega}|_{{\mathcal H}(H_{\omega},R(V))} \;\; \mbox{has simple point spectrum for a.e.\ $\omega$},
\end{equation}
and
\begin{equation} \label{eq:weakcyclic}
{\mathcal H}^{pp}(H_{\omega}) \subset {\mathcal H}(H_{\omega}, R(V)) \;\; \mbox{for a.e.\ $\omega$}.
\end{equation}

Noting (\ref{eq:lambdaind}), (\ref{eq:BSsimple}) can be established via Theorem~\ref{thm:criterion} if we can verify simplicity of the boundary values of the Birman-Schwinger operator $\sqrt{V}(H_{\hat{\omega}}-z)^{-1} \sqrt{V}$ for almost every $\hat{\omega}$. We will thus refer to (\ref{eq:BSsimple}) as {\it simplicity of the Birman-Schwinger operators}. In addition, we will have to establish (\ref{eq:weakcyclic}) which we will refer to as {\it weak cyclicity of $R(V)$} ({\it cyclicity} of $R(V)$ denotes the stronger property that ${\mathcal H}(H_{\omega},R(V)) = {\mathcal H}$. In the language used in the introduction, (\ref{eq:BSsimple}) reflects the absence of local symmetries in the model, while (\ref{eq:weakcyclic}) can be interpreted as absence of global symmetries.

It's quite enlightening so compare the discrete Anderson model (\ref{disAnderson}) and Models A, B and C from the point of view of differences which arise when trying to verify (\ref{eq:BSsimple}) and (\ref{eq:weakcyclic}).

The discrete Anderson model (\ref{disAnderson}) and Model C, the continuum Anderson model, represent two extreme cases. For the discrete Anderson model the boundary values $G(E+i0)$ are rank-one operators and thus trivially have simple non-zero eigenvalues (existence of $G(E+i0)$ for almost every $E$ holds for all three models as discussed in the Appendix). Thus for the discrete Anderson model only weak cyclicity of $R(V)$ (in this case the span of the canonical basis vector $e_0$) needs to be checked. This is essentially what was done in \cite{Simon}, whose arguments can be traced in our discussion in Section~\ref{sec:simon} (and are a special case of the result shown there).

The situation for the continuum Anderson model is reversed. In this case the weak cyclicity of $R(V) = \{f\phi: \phi \in L^2(\R^d)\}$ is well known. In fact, under the additional assumption that $f>0$ on a non-trivial open set it is known that $R(V)$ is cyclic for every Schr\"odinger operator $H=-\Delta+q$ with, say, bounded potential $q$. This is a consequence of unique continuation for Schr\"odinger operators, see e.g.\ \cite{Combes/Hislop}. However, proving simplicity of the non-zero eigenvalues of the infinite rank Birman-Schwinger operator $f^{1/2}(H_{\hat{\omega}}-(E+i0))^{-1} f^{1/2}$ for almost every $\hat{\omega}$ is a hard problem for the continuum Anderson model, reflecting the very rich structure of possible local symmetries in the continuum. In this paper we will have nothing to say about this.

Instead we will focus on verifying (\ref{eq:BSsimple}) and (\ref{eq:weakcyclic}) for Models A and B. For these models both conditions are non-trivial, but, due to the finite rank property of $V$, technically accessible with linear algebra tools. Checking weak cyclicity of $R(V)$ is non-trivial for these models because of the lack of a general unique continuation property for discrete Schr\"odinger equations. A look at our proofs of cyclicity for Models A and B shows that they can be interpreted as salvaging analogs of unique continuation properties in some specific situations.

For Model A we will be able to verify (\ref{eq:BSsimple}) and (\ref{eq:weakcyclic}) and thus prove Theorem~\ref{thm:modelA} in full generality, only requiring the necessary condition that the single site matrix $W$ has simple eigenvalues. For Model B we can prove weak cyclicity of $R(V)$ without further restrictions, but can show simplicity of the Birman-Schwinger operators only for some special cases, leading to Theorem~\ref{thm:modelB}.
One may think of Model B as a discretization of Model C. For a discretization to provide a good approximation of the continuum one needs to choose a fine mesh corresponding to large $C_0$. Thus, trying to consider Model B with large $C_0$ allows to anticipate the difficulties, due to an increasing number of local symmetries, in aiming at ultimately handling the continuum Anderson model. Our proofs will shed some light on this.

\section{A generalization of Simon's argument} \label{sec:simon}

For the rest of this paper we will consider Models A and B only. In both cases the single site potential is a finite rank perturbation, which allows to use the following extension of an argument from \cite{Simon}.

\begin{theorem} \label{thm:simon2}
Let $H$ be self-adjoint in the separable Hilbert space $\mathcal H$, $k\in \N$, and $X$ and $Y$ $k$-dimensional subspaces of $\mathcal H$ with orthogonal projections $P_X$ and $P_Y$. Let $V\ge 0$ with $R(V)=X$ and $H_{\lambda} := H+\lambda V$, $\lambda \in \R$.

Suppose that there exists $z_0 \in \C^+$ such that
\begin{equation} \label{eq:spanX}
\mbox{\rm{span}}\{ R(P_X(H_{\mu}-z_0)^{-1} P_Y)| \;\mu \in \R\} = X.
\end{equation}
Then for Lebesgue-a.e.\ $\lambda\in \R$,
\begin{equation} \label{eq:pppxy}
P^{pp}(H_{\lambda}) X \subset P^{pp}(H_{\lambda}) Y.
\end{equation}
\end{theorem}

Observe that (\ref{eq:pppxy}) implies $P^{pp}(H_{\lambda}) X \subset {\mathcal H}(H_{\lambda},Y)$, which is how it will be used in verifying (\ref{eq:weakcyclic}) for Models A and B below. There we will also use that the Hilbert space can be spanned by subspaces $X$ on which (\ref{eq:pppxy}) holds.

There are two special cases of Theorem~\ref{thm:simon2} worthwhile mentioning: (i) If $P_X (H-z_0)^{-1} I_Y$ is invertible from $X$ to $Y$ (and thus surjective) for one $z_0\in \C^+$, then (\ref{eq:spanX}) trivially holds, requiring only the use of the single coupling constant $\mu=0$. Below we will verify this form of the condition for Model A. (ii) In the rank one case dim$\,X =$ dim$\,Y=1$, i.e.\ the discrete Anderson model (\ref{disAnderson}), the previous special case means non-vanishing of the matrix-element $(H-z_0)^{-1}(x,y)$ for one $z_0\in \C^+$. This is how the argument behind Theorem~\ref{thm:simon2} enters in \cite{Simon}.

\begin{proof}(of Theorem~\ref{thm:simon2})
First note that due to the finite dimension of $X$, there is a finite set $N\subset \R$ (of at most dim$\,X$ elements) such that it suffices to take the span over $\mu\in N$ on the left hand side of (\ref{eq:spanX}).

Let $A$ be the set of all those $t\in \R$ for which either $P_X(H_{\mu}-(t+i0))^{-1} P_Y$ does not exist for at least one $\mu\in N$ or such that
\begin{equation} \label{eq:spanbv}
\mbox{\rm{span}} \{ R(P_X(H_{\mu}-(t+i0))^{-1} P_Y)| \;\mu \in N\}
\end{equation}
is not all of $X$.

We see that $|A|=0$ as follows: For $\phi$ in the Hilbert space and fixed $\mu$, $((H_{\mu}-z)^{-1}\phi,\phi)$ is Herglotz as a function of $z\in \C^+$. Thus, by polarization and Lemma~\ref{lem:bc}, $((H_{\mu}-z)^{-1}\phi,\psi)$ is of bounded characteristic for arbitrary $\phi$ and $\psi$. Thus the matrices $P_X(H_{\mu}-z)^{-1} P_Y$, $\mu\in N$, represented with respect to fixed orthonormal bases in $X$ and $Y$, have entries of bounded characteristic. By Lemma~\ref{lem:green}(a) the boundary values $P_X(H_{\mu}-(t+i0))^{-1} P_Y$ exist for all $\mu\in N$ and a.e.\ $t\in \R$.

Furthermore, by (\ref{eq:spanX}) there is a collection of $k$ columns $f_j(z)$, $j=1,\ldots,k$, chosen from the matrices $P_X(H_{\mu}-z)^{-1} P_Y$, $\mu\in N$, such that the square matrix $(f_1(z), \ldots, f_k(z))$ is invertible at $z=z_0$. Thus by Lemma~\ref{lem:green}(b) we see that the boundary value $(f_1(t+i0),\ldots, f_k(t+i0))$ is invertible for almost every $t$, showing that (\ref{eq:spanbv}) gives all of $X$. This completes the proof of $|A|=0$.

Therefore $A' := A \cup \left( \bigcup_{\mu\in N} \sigma_p(H_{\mu}) \right)$ is a nullset as well. By Lemma~\ref{specaverage} there is a nullset $M \subset \R$ such that
\begin{equation} \label{eq:noevals}
H_{\lambda}|_{{\mathcal H}_V} \:\mbox{has no eigenvalues in $A'$ for all $\lambda \in
\R\setminus M$}.
\end{equation}

Fix $\lambda \in \R \setminus (M\cup N)$ and write
\begin{equation} \label{eq:efexp}
P^{pp}(H_{\lambda})X=\sum_{e\in \sigma_{p}(H_{\lambda})}P_e^{\lambda}X
\end{equation}
with $P_e^{\lambda}:= \chi_{\{e\}}(H_{\lambda})$, the eigenspace of $H_{\lambda}$ to $e$.

For now fix $e \in
\sigma_p(H_{\lambda})$ with $P_e^{\lambda}X\neq\{0\}$ and also fix $\mu \in N$. Thus $e$ is an eigenvalue of $H_{\lambda}|_{{\mathcal H}_V}$,
which by (\ref{eq:noevals}) can not lie in $A'$, in particular, $e \notin \sigma_p(H_{\mu})$.
As $H_{\lambda}$ is obtained via a rank $k$ perturbation of $H_{\mu}$ we have
\[
\mbox{rank} \ P_e^{\lambda}\leq k<\infty
\]

By the Second Resolvent Identity,
\[
(H_{\lambda}-z)^{-1}=(H_{\mu}-z)^{-1}-(\lambda-\mu)
(H_{\lambda}-z)^{-1}V(H_{\mu}-z)^{-1}, \quad z \in \mathbb{C}\setminus
\mathbb{R},
\]
and, in particular, for $\epsilon>0$
\begin{eqnarray*}
i\epsilon P_e^{\lambda}(H_{\lambda}-e-i\epsilon)^{-1}P_Y&=&i\epsilon
P_e^{\lambda}(H_{\mu}-e-i\epsilon)^{-1}P_Y-{} \\
&& {}-(\lambda-\mu) P_e^{\lambda}i\epsilon (H_{\lambda}-e-i\epsilon)^{-1} V
(H_{\mu}-e-i\epsilon)^{-1}P_Y.
\end{eqnarray*}
Since $e\not\in \sigma_p(H_{\mu})$, letting $\epsilon \downarrow 0$ in the last equality
gives
\[
P_e^{\lambda}P_Y=-(\lambda-\mu) P_e^{\lambda} V(H_{\mu}-(e+i0))^{-1} P_Y,
\]
where we have used the fact that
$P_e^{\lambda}=\lim_{\epsilon \downarrow
0}i\epsilon(H_{\lambda}-e-i\epsilon)^{-1}$ in the weak sense (and thus in norm due to finite dimension).
We infer
\[ P_e^{\lambda} R(V(H_{\mu}-(e+i0))^{-1} P_Y) \subset P_e^{\lambda} Y \quad \mbox{for all $\mu\in N$}.\]
But, using $R(V)=X$ and $e\not\in A$,
\begin{eqnarray*}
\lefteqn{\mbox{\rm{span}} \{ R(V(H_{\mu} - (e+i0))^{-1} P_Y): \;\mu\in N\}} \\
& = & \mbox{\rm{span}} \{ R(P_X(H_{\mu}-(e+i0))^{-1} P_Y):\; \mu\in N\} = X,
\end{eqnarray*}
which yields
\[ P_e^{\lambda} X \subset P_e^{\lambda} Y \subset P^{pp}(H_{\lambda}) Y.\]
This holds for every $e\in \sigma_p(H_{\lambda})$ with $P_e^{\lambda} X \not= \{0\}$, so (\ref{eq:efexp}) implies (\ref{eq:pppxy}).
\end{proof}

\section{Model A} \label{sec:modelA}

As a first application of the general theory developed so far, we will now prove Theorem~\ref{thm:modelA} by verifying (\ref{eq:BSsimple}) and (\ref{eq:weakcyclic}).

For the duration of this proof we write $H=H_{\omega}=H_{\omega}^A$. For $|z|> \|H\|$ we have the Neumann series
\begin{equation} \label{eq:Neumann}
(H-z)^{-1} = -\frac{1}{z} (I+ \frac{1}{z}H + \frac{1}{z^2} H^2 + \ldots),
\end{equation}
which will be used in the verification of both, (\ref{eq:weakcyclic}) and (\ref{eq:BSsimple}). We start with the latter.

Here (\ref{eq:Neumann}) implies that
\[ V_0^{1/2} (H-z)^{-1} V_0^{1/2} = -\frac{1}{z}(V_0 + O(1/|z|)). \]
As $V_0|_{R(V_0)} = W$, this shows that simplicity of $W$ leads to simplicity of
\begin{equation} \label{eq:C2}
V_0^{1/2} (H_{\omega}-z)^{-1} V_0^{1/2}: R(P_0) \to R(P_0)
\end{equation}
for all $\omega$ and $z\in \C^+$ with $|z|$ sufficiently large.

As in Section~\ref{sec:abstract} write $\omega = (\hat{\omega}, \omega_0)$ where $\hat{\omega} = (\omega_n)_{n\not= 0}$. As observed in the proof of Theorem~\ref{thm:simon2} the entries of $V_0^{1/2}(H_{\omega}-z)^{-1} V_0^{1/2}$ are of bounded characteristic. It follows from (\ref{eq:C2}) and Lemma~\ref{lem:green}(c) that for every $\hat{\omega}$ the boundary value $V_0^{1/2} (H_{\hat{\omega}}-(t+i0))^{-1} V_0^{1/2}$ exists and has simple non-zero eigenvalues for almost every $t\in \R$.

Thus we can apply Theorem~\ref{thm:criterion} with $H_0 = H_{\hat{\omega}}$ and $V=V_0$. Using that $\omega_0$ has absolutely continuous distribution, we conclude that (\ref{eq:BSsimple}) holds for almost every $\omega$.

The following argument to verify the assumption of Theorem~\ref{thm:simon2} and thus prove (\ref{eq:weakcyclic}) is essentially found in \cite{Simon}. For $j\not= 0$ we get from (\ref{eq:Neumann}) that
\[ P_j (H-z)^{-1} P_0 = -\frac{1}{z^2} P_j H P_0 - \frac{1}{z^3} P_j H^2 P_0 + \ldots.\]
Using that $H$ has only next neighbor hopping terms of magnitude one and letting $|j|=|j_1|+\ldots+|j_d|$ for $j\in\Z^d$, we observe that
\[ P_j H^{\ell} P_0 = 0 \quad \mbox{if $|j|<\ell$}, \]
and
\[ P_j H^{\ell} P_0 = P_j h_0^{\ell} P_0 = C_{j,d}I \quad \mbox{if $|j|=\ell$},\]
where the latter is viewed as an operator from $R(P_0)$ to $R(P_j)$ (with $I$ the matrix representation in the canonical bases of these spaces) and $C_{j,d} \not= 0$ is the number of shortest paths from $0$ to $j$ in $\Z^d$. We conclude that
\[ P_j (H-z)^{-1} P_0 = - \frac{C_{j,d}}{z^{\ell+1}} I + O(1/|z|^{\ell+2}),\]
which is invertible for $|z|$ sufficiently large.

Theorem~\ref{thm:simon2} and a re-sampling argument in the absolutely continuous random variable $\omega_j$ yields that
\[ P^{pp}(H_{\omega}) R(P_j) \subset {\mathcal H}(H_{\omega},R(P_0)) \;\mbox{for a.e.\ $\omega$}. \]
Note that this also holds trivially for $j=0$. Using that $\{R(P_j): j\in \Z^d\}$ spans ${\mathcal H}$ and taking a countable intersection of full measure sets we conclude that
\[ {\mathcal H}^{pp}(H_{\omega}) = P^{pp}(H_{\omega}){\mathcal H} \subset {\mathcal H}(H_{\omega},R(P_0)) \;\mbox{for a.e.\ $\omega$}, \]
which is (\ref{eq:weakcyclic}).

\section{Weak cyclicity for Model B} \label{sec:weakcyclicModB}

The goal of this section is to verify weak cyclicity (\ref{eq:weakcyclic}) for Model B, which we can do in full generality, i.e.\ for any choice of $C_0 = \{0,\ldots,L_1-1\} \times \ldots \times \{0,\ldots, L_d-1\}$ and $f:C_0 \to (0,\infty)$.

\begin{proposition} \label{prop:C1primecrit}
Model B satisfies
\[ {\mathcal H}^{pp}(H_{\omega}^B) \subset {\mathcal H}(H_{\omega}^B, R(P_0)) \]
for almost every $\omega$.
\end{proposition}

We first prove a lemma. For any tile $C=C_m$, $m\in \Z^d$, we denote by $f_C$ the single-site potential on $C$, i.e.\ $f_C = f_m = f(\cdot-mL)$. A neighboring tile $C'=C_{m'}$ of $C$ is a tile such that $m$ and $m'$ coincide in all but one coordinate, and differ by $1$ in the latter. For a pair $(C,C')$ of neighboring tiles $h_0^{(C,C')}$ is the restriction of $h_0$ to $\ell^2(C\cup C')$.

\begin{lemma} \label{lem:2tilespan}
For every $z_0\in \C^+$ and every pair of neighboring tiles $(C,C')$ it holds that
\begin{equation} \label{eq:2tilespan}
\mbox{\rm{span}} \{ R(\chi_C (h_0^{(C,C')}+\mu f_C -z_0)^{-1} \chi_{C'})| \; \mu \in \R\} = \ell^2(C).
\end{equation}
\end{lemma}

\begin{proof}
Without loss of generality we may assume that $C=C_0$ and $C'=C_{(-1,0,\ldots,0)}$. We have to show that
\begin{eqnarray*}
\lefteqn{\bigcap_{\mu} (R(\chi_C (h_0^{(C,C')}+\mu f_C - z_0)^{-1} \chi_{C'}))^{\perp}} \\
& = & \bigcap_{\mu} N(\chi_{C'} (h_0^{(C,C')}+\mu f_C -\bar{z}_0)^{-1} \chi_C ) = \{0\}.
\end{eqnarray*}
Thus assume that $\psi \in \ell^2(C)$ is such that
\begin{equation} \label{eq:psicond}
\chi_{C'} (h_0^{(C,C')}+\mu f_C - \bar{z}_0)^{-1} \psi = 0 \quad \mbox{for all $\mu\in \R$}.
\end{equation}
Our goal is to show that $\psi=0$.

Let $g_{\mu} := (h_0^{(C,C')}+\mu f_C - \bar{z}_0)^{-1} \psi$. Then supp$\,g_{\mu} \subset C$ and
\begin{equation} \label{eq:gcond}
\psi = (h_0^{(C,C')}+\mu f_C -\bar{z}_0) g_{\mu}
\end{equation}
for all $\mu\in \R$. For $\psi$, $g_{\lambda}$ and $f_C$, all supported on $C$, we will consider transversal sections corresponding to fixed value of the first coordinate, i.e.\ we write
\[ \psi_k(n_2,\ldots,n_d) = \psi(k,n_2,\ldots,n_d),  \]
and similar $g_{\mu,k}$, $f_{C,k}$ for the sections of $g_{\mu}$ and $f_C$. By $h_0^1$ we denote the restriction of the $d-1$-dimensional discrete Laplacian to $\{0,\ldots,L_2-1\} \times \ldots \times \{0,\ldots,L_d-1\}$.

Evaluating (\ref{eq:gcond}) at value $k=-1$ of the first coordinate gives
\begin{equation} \label{eq:k-1}
0=g_{\mu,0}.
\end{equation}
If $L_1=1$, this means $g_{\mu}=0$ and thus $\psi=0$ by (\ref{eq:gcond}). If $L_1>1$, we evaluate (\ref{eq:gcond}) at values $0\le k \le L_1-1$ to get
\begin{equation} \label{eq:k0}
\psi_0 = g_{\mu,1} + (h_0^1 +\mu f_{C,0}-\bar{z}_0) g_{\mu,0},
\end{equation}
\begin{equation} \label{eq:k1}
\psi_k = g_{\mu,k-1} + g_{\mu,k+1} + (h_0^1 +\mu f_{C,k}-\bar{z}_0) g_{\mu,k}, \quad 1\le k \le L_1-2,
\end{equation}
\begin{equation} \label{eq:k2}
\psi_{L_1-1} = g_{\mu,L_1-2} + (h_0^1+\mu f_{C,L_1-1} -\bar{z}_0) g_{\mu,L_1-1}.
\end{equation}

Inserting (\ref{eq:k-1}) into (\ref{eq:k0}) and then, successively for $1\le k\le L_1-2$, into (\ref{eq:k1}) yields
\[ g_{\mu,k} = (-1)^{k-1} \mu^{k-1} f_{C,1} \ldots f_{C,k-1} \psi_0 + O(\mu^{k-2}) \]
as $\mu\to\infty$ for all $k\le L_1-1$. Ultimately, inserting into (\ref{eq:k2}) gives
\[ \psi_{L_1-1} = (-1)^{L_1-1} \mu^{L_1-1} f_{C,1} \ldots f_{C,L_1-1} \psi_0 + O(\mu^{L_1-2}).\]
As this must hold for all $\mu$, we conclude that $\psi_0=0$, and thus, by (\ref{eq:k0}), $g_{\mu,1}=0$.

This allows to reinterpret (\ref{eq:psicond}) as
\begin{equation} \label{eq:psicond2}
\chi_{C_+'} (h_0^{(C,C')}+\mu f_C-\bar{z}_0)^{-1} \psi =0 \quad \mbox{for all $\mu\in \R$}
\end{equation}
and $\psi \in \ell^2(C_-)$. Here $C_+'$ and $C_-$ are the boxes found by moving the left-most layer of $C$ to $C'$, i.e.\
\[ C_- := \{1,\ldots,L_1-1\} \times \{0,\ldots,L_2-1\} \times \ldots \times \{0,\ldots, L_d-1\},\]
\[ C_+' := \{-L_1,\ldots,0\} \times \{0,\ldots,L_2-1\} \times \ldots \times \{0,\ldots, L_d-1\}.\]

This shows that the process of calculating (\ref{eq:gcond}) to (\ref{eq:k2}) can be repeated, now starting with (\ref{eq:psicond2}) and leading to $\psi_1=0$, $g_{\mu,2}=0$. Iterating we find $\psi_0=\psi_1= \ldots = \psi_{L_1-1}=0$, and thus $\psi=0$.

\end{proof}

For the remainder of this and the following section we write $H=H_{\omega}= H_{\omega}^B$.

\begin{proof}(of Proposition~\ref{prop:C1primecrit})
For an arbitrary pair of neighboring tiles $(C,C')$ we can apply Lemma~\ref{lem:2tilespan}.
As $\ell^2(C)$ is finite-dimensional there exists a finite set $N\subset \R$ such that it suffices to take the span over $\mu\in N$ on the left hand side of (\ref{eq:2tilespan}). Let $C=C_m$, $C'=C_{m'}$, $\Lambda_B$ a boundary layer consisting of tiles enclosing $C\cup C'$, $\Lambda_{ext} = \Z^d \setminus \{ C \cup C' \cup \Lambda_B\}$, and for given $\omega = (\omega_{\ell})_{\ell\in \Z^d}$ define
\begin{equation} \label{eq:choice}
\omega_{\ell}^{(\lambda,\mu)} = \left\{ \begin{array}{ll} \mu, & \ell = m,\\ 0, & \ell = m', \\ \lambda, & \mbox{on sites in $\Lambda_B$}, \\ \omega_{\ell}, & \mbox{on sites in $\Lambda_{ext}$}. \end{array} \right.
\end{equation}
We claim that, for fixed $\mu$,
\begin{equation} \label{eq:couplinglimit}
\lim_{\lambda\to\infty} \chi_C (H_{\omega^{(\lambda,\mu)}} -z_0)^{-1} \chi_{C'} = \chi_C (h_0^{(C,C')}+\mu f_C-z_0)^{-1} \chi_{C'}.
\end{equation}
This is shown by a Schur complementation argument:
Decompose $\ell^2(\Z^d) = \ell^2(\Z^d \setminus \Lambda_B) \oplus \ell^2(\Lambda_B)$ and let $P = P_{\Z^d\setminus \Lambda_B} = P_{C\cup C'} \oplus P_{\Lambda_{ext}}$ and $Q= I-P = P_{\Lambda_B}$ be the corresponding orthogonal projections. Write
\[ H_{\omega^{(\lambda,\mu)}} - z_0 = \left( \begin{array}{cc} A & B \\ C & D \end{array} \right) \]
as a block operator with respect to this decomposition. Here $D= Q(h_0-z+\lambda \sum_j V_j)Q$ in $\ell^2(\Lambda_B)$. As $\sum_j V_j$ has a uniform positive lower bound, $D$ is invertible for $\lambda$ sufficiently large and $\lim_{\lambda\to\infty} D^{-1} =0$. $A$, $B$ and $C$ do not depend on $\lambda$.

Thus Schur complementation yields
\[ P(H_{\omega^{(\lambda,\mu)}}-z_0)^{-1} P = (A-BD^{-1}C)^{-1} \to A^{-1} \quad \mbox{as $\lambda\to\infty$}.\]
We also have
\[ A = P_{C\cup C'} (h_0+\mu f_C-z_0)P_{C\cup C'} + P_{\Lambda_{ext}} (H_{\omega}-z_0) P_{\Lambda_{ext}}, \]
giving
\begin{eqnarray*}
\chi_{C} (H_{\omega^{(\lambda,\mu)}}-z_0)^{-1} \chi_{C'} & = & \chi_{C} P_{C\cup C'} (H_{\omega^{(\lambda,\mu)}}-z_0)^{-1} P_{C\cup C'} \chi_{C'} \\
& \to & \chi_{C} (h_0^{(C,C')}+\mu f_C-z_0)^{-1} \chi_{C'}
\end{eqnarray*}
as $\lambda \to \infty$, proving (\ref{eq:couplinglimit}).

Finiteness of $N$ implies the existence of $\lambda_0$ sufficiently large such that
\begin{eqnarray*}
\lefteqn{\mbox{\rm{span}} \{R(\chi_C (H_{\omega^{(\lambda_0,\mu)}}-z_0)^{-1} \chi_{C'}) |\; \mu\in N\}} \\
& = & \mbox{\rm{span}} \{R(\chi_C (h_0^{(C,C')}+\mu f_C-z_0)^{-1} \chi_{C'}) |\; \mu\in N\} \\
& = & \ell^2(C).
\end{eqnarray*}
Picking dim$\,\ell^2(C) = |C|$ linearly independent columns of the matrices $\{ \chi_C (H_{\omega^{(\lambda_0,\mu)}}-z_0)^{-1} \chi_{C'}|\; \mu\in N\}$, we observe that the determinant of the matrix formed by these columns is analytic in each of the parameters $\omega_{\ell}$ corresponding to sites in $\Lambda_B \cup C'$. As the determinant is non-zero for the special choice made in (\ref{eq:choice}), we can successively use analyticity in these parameters to conclude that for almost every $\hat{\omega} = (\omega_{\ell})_{\ell \not= m}$,
\[ \mbox{\rm{span}} \{ R (\chi_C (H_{\hat{\omega}}+\mu f_C-z_0)^{-1} \chi_{C'}) |\; \mu\in N \} = \ell^2(C), \]
where $H_{\hat{\omega}} = h_0 + \sum_{\ell \not= m} \omega_{\ell} f_{\ell}$. We have thus verified the assumptions of Theorem~\ref{thm:simon2} with $X= \ell^2(C)$, $Y=\ell^2(C')$, $H=H_{\hat{\omega}}$ and $V=f_C$ and can therefore conclude that
\[ P^{pp}(H_{\omega}) \ell^2(C) \subset P^{pp}(H_{\omega}) \ell^2(C') \quad \mbox{for almost every $\omega$}.\]

As this holds for any pair of neighboring tiles $(C,C')$, we may iterate to conclude
\[ P^{pp}(H_{\omega}) \ell^2(C_n) \subset P^{pp}(H_{\omega}) \ell^2(C_m) \quad \mbox{for any $(n,m)$ and a.e.\ $\omega$}.\]
Finally, choosing $m=0$ and taking the union over $n$ on the left, we get
\[ {\mathcal H}^{pp}(H_{\omega}) \subset P^{pp}(H_{\omega}) \ell^2(C_0) \subset {\mathcal H}(H_{\omega}, \ell^2(C_0)) \quad \mbox{for a.e.\ $\omega$},\]
as was to be shown.

\end{proof}

\section{Simplicity of the Birman-Schwinger operator for Model B} \label{sec:BSsimpleModB}

This final section is aimed at verifying (\ref{eq:BSsimple}) for Model B, that is simplicity of the restriction of $H=H_{\omega}=H_{\omega}^B$ to the reducing subspace generated by the single site potential. We will accomplish this via Theorem~\ref{thm:criterion} by showing simplicity of the corresponding Birman-Schwinger operators. It is here where we don't have a general argument and will have to use one of the additional conditions given in Theorem~\ref{thm:modelB}. As discussed in Section~\ref{sec:abstract}, when combined with Proposition~\ref{prop:C1primecrit} this completes the proof of Theorem~\ref{thm:modelB}.

\begin{proposition} \label{pro:BC2}
Suppose that for Model B one of the additional assumptions (i), (ii) or (iii) in Theorem~\ref{thm:modelB} holds. Then $H_{\omega}^B|_{{\mathcal H}(H_{\omega}^B, R(V))}$ has simple point spectrum for almost every $\omega$.
\end{proposition}

We will prove this by establishing that for almost every $\omega$ there exists $z\in \C^+$ such that
\begin{equation} \label{eq:modelBC2}
\sqrt{f} (H_{\omega}-z)^{-1} \sqrt{f}: \ell^2(C_0) \to \ell^2(C_0)
\end{equation}
is simple. In fact, under conditions (i) or (ii) of Theorem~\ref{thm:modelB} this will hold deterministically, i.e.\ for every $\omega$, but in case of (iii) we only get an almost sure result. Based on this and Theorem~\ref{thm:criterion}, Proposition~\ref{pro:BC2} now follows with the same argument which was used for Model A in Section~\ref{sec:modelA}. Here it suffices to know that (\ref{eq:modelBC2}) holds almost surely.

As in Section~\ref{sec:modelA} for Model A, our argument starts with the Neumann series (\ref{eq:Neumann}). The easiest case is (i), i.e.\ simplicity of $f$, in which case using only the first order approximation in (\ref{eq:Neumann}) gives
\begin{equation} \label{eq:casei}
-z\sqrt{f}(H-z)^{-1} \sqrt{f} = f+O(1/|z|)
\end{equation}
as an operator in $\ell^2(C_0)$ and for $|z|\to\infty$. Since $f$ is simple it follows that $\sqrt{f} (H-z)^{-1} \sqrt{f}$ is simple for sufficiently large $|z|$.

Now consider the case (ii), $C_0 =\{0\} \times \{0,\ldots,L-1\}$, $f=\chi_{C_0}$. In this case $\sqrt{f} = \chi_{C_0}$ and we use the second order approximation in (\ref{eq:Neumann}) to conclude that
\begin{equation} \label{eq:caseii1}
-z\chi_{C_0} (H-z)^{-1} \chi_{C_0} = \chi_{C_0} +\frac{1}{z} \chi_{C_0} h_0 \chi_{C_0} + \frac{\omega_0}{z} \chi_{C_0} + O(1/|z|^2)
\end{equation}
and thus
\begin{equation} \label{eq:caseii2}
z\left( -z \chi_{C_0} (H-z)^{-1} \chi_{C_0} - (1+\frac{\omega_0}{z}) \chi_{C_0}\right) = \chi_{C_0} h_0 \chi_{C_0} + O(1/|z|)
\end{equation}
as $|z|\to\infty$. In the canonical basis of $\ell^2(C_0)$,
\[ \chi_{C_0} h_0 \chi_{C_0} = \left( \begin{array}{cccc} 0 & 1 & &  \\ 1 & \ddots & \ddots &  \\  & \ddots & \ddots & 1 \\ & & 1 & 0 \end{array} \right), \]
a finite Jacobi matrix with simple eigenvalues. Thus the left hand side of (\ref{eq:caseii2}) and therefore $\chi_{C_0} (H-z)^{-1} \chi_{C_0}$, is simple for $|z|$ sufficiently large.

So far our arguments can be summarized as follows: In case (i) the first term in the asymptotic expansion (\ref{eq:Neumann}) suffices to break all degeneracies. For case (ii) the degeneracies are broken by the second term in the expansion.

Case (iii) is considerably more complicated. We will have to explicitly calculate several more terms in the asymptotic expansion. Degeneracies will not be broken completely by including the next term in the series, but only partly. Thus we have to control the effect of terms of different orders on eigenvalues carefully, to avoid that eigenvalues which are split by lower order terms become degenerate again by adding higher order terms to split the remaining degeneracies.

Instead of the full random operator $H_{\omega}$ we will start by considering the operator $h_{a,b}=h_0 + V_{a,b}$ with potential restricted to two sites,
\[ V_{a,b}(j) := a \chi_{C_0}(j_1,j_2-2) + b\chi_{C_0}(j_1+2,j_2) \]
for all $j=(j_1,j_2) \in \Z^2$. Thus supp$\,V_{a,b} = C_{(0,1)} \cup C_{(-1,0)}$, only the sites above and to the left of $C_0$ are occupied.

\begin{proposition} \label{pro:ab}
If $a\not= b$, then $\chi_{C_0}(H_{a,b}-z)^{-1} \chi_{C_0}$ as an operator on $\ell^2(C_0)$ has simple eigenvalues for $|z|$ sufficiently large.
\end{proposition}

\begin{proof}
For matrix-representations of operators in $\ell^2(C_0)$ we will throughout use the following orthonormal basis, which is best suitable to reflect the various symmetries in the model (and which need to be broken):
\begin{equation} \label{eq:basis}
\delta_1 = \frac{1}{2} \left( \begin{array}{rr} 1 & 1 \\ 1 & 1 \end{array} \right), \delta_2 = \frac{1}{2} \left( \begin{array}{rr} 1 & -1 \\ -1 & 1 \end{array} \right), \delta_3 = \frac{1}{2} \left( \begin{array}{rr} 1 & 1 \\ -1 & -1 \end{array} \right), \delta_4 = \frac{1}{2} \left( \begin{array}{rr} 1 & -1 \\ 1 & -1 \end{array} \right),
\end{equation}
where we represent functions on $C_0$ as $2\times 2$-arrays.

We have $\chi_{C_0} h_{a,b} = \chi_{C_0} h_0$ and $h_{a,b} \chi_{C_0} = h_0 \chi_{C_0}$. Thus the expansion (\ref{eq:Neumann}) written down up to fourth order yields, as an operator in $\ell^2(C_0)$,
\begin{eqnarray} \label{eq:caseiii1}
\lefteqn{z(-z\chi_{C_0} (h_{a,b}-z)^{-1} \chi_{C_0} -I)} \\ \nonumber
& = & \chi_{C_0} h_0 \chi_{C_0} + \frac{1}{z} \chi_{C_0} h_0^2 \chi_{C_0} + \frac{1}{z^2} h_0(h_0+V_{a,b})h_0\chi_{C_0} \\ \nonumber
& & \mbox{} + \frac{1}{z^3} h_0(h_0+V_{a,b})^2 h_0\chi_{C_0} + O(1/|z|^4).
\end{eqnarray}

The calculation of the various terms on the right hand side of (\ref{eq:caseiii1}) should be done geometrically, starting from the arrays giving the vectors $\delta_i$, $i=1,\ldots,4$, and using that $h_0$ acts on every two-dimensional array of numbers by adding up all neighboring values at each site.

With considerable effort we get
\[
\chi_{C_0} h_0 \chi_{C_0}  =  \left( \begin{array}{cccc} 2 & 0 & 0 & 0 \\ 0 & -2 & 0 & 0 \\ 0 & 0 & 0 & 0 \\ 0 & 0 & 0 & 0 \end{array} \right), \]
\[
\chi_{C_0} h_0^2 \chi_{C_0} = \left( \begin{array}{cccc} 6 & 0 & 0 & 0 \\ 0 & 6 & 0 & 0 \\ 0 & 0 & 2 & 0 \\ 0 & 0 & 0 & 2 \end{array} \right), \]
\[
\chi_{C_0} h_0^3 \chi_{C_0} = \left( \begin{array}{cccc} 18 & 0 & 0 & 0 \\ 0 & -18 & 0 & 0 \\ 0 & 0 & 0 & 0 \\ 0 & 0 & 0 & 0 \end{array} \right), \]
\[
\chi_{C_0} h_0V_{a,b} h_0 \chi_{C_0} = \left( \begin{array}{cccc} a+b & 0 & a & b \\ 0 & a+b & b & a \\ a & b & a+b & 0 \\ b & a & 0 & a+b \end{array} \right), \]
\[
\chi_{C_0} h_0 (h_0+V_{a,b})^2 h_0 \chi_{C_0} = (12+\frac{a^2+b^2}{2})I + \left( \begin{array}{cccc} * & * & * & * \\ * & * & * & * \\ * & * & a-b & 0 \\ * & * & 0 & b-a \end{array} \right). \]
For the latter matrix we will only need the lower right $2\times 2$-block.

Thus, suppressing constant multiples of the $4\times 4$-identity matrix, we find that it suffices to show simplicity of the $2\times 2$-block matrix
\[ g_{a,b}(z) = \left( \begin{array}{cc} A(z) & B(z) \\ C(z) & D(z) \end{array} \right), \]
where
\[ A(z) = \left( \begin{array}{rr} 2 & 0 \\ 0 & -2 \end{array} \right) + O(1/|z|), \]
\[ B(z) = C(z) = \frac{1}{2z^2} B_0 + O(1/|z|^3), \quad B_0 := \left( \begin{array}{cc} a & b \\ b & a \end{array} \right), \]
\[ D(z) = \frac{1}{z^3} D_0 + O(1/|z|^4), \quad D_0 := \left( \begin{array}{cc} a-b & 0 \\ 0 & b-a \end{array} \right). \]

For $|z|$ sufficiently large, $g_{a,b}$ has one eigenvalue each near $2$ and $-2$ and two eigenvalues (counted with multiplicity) near $0$. The latter two eigenvalues satisfy $\lambda = O(1/|z|)$ and we must show that they are distinct. For each of these eigenvalues $A-\lambda I$ is invertible and we can therefore use Schur complementation to find the corresponding eigenvectors: Suppose that
\[ \left( \begin{array}{cc} A-\lambda I & B \\ C & D-\lambda I \end{array} \right) \left( \begin{array}{c} \phi_1 \\ \phi_2 \end{array} \right) = 0.\]
Then $(A-\lambda I)\phi_1 + B\phi_2 =0$ and $C\phi_1+(D-\lambda I)\phi_2 =0$ and we can eliminate $\phi_1$ to get
\[ -C(A-\lambda I)^{-1} B\phi_2 +(D-\lambda I)\phi_2 =0.\]
The two eigenvalues of $g_{a,b}$ with $\lambda=O(1/|z|)$ are therefore roots of
\begin{equation} \label{eq:schurdet}
\det (D-\lambda I -C(A-\lambda I)^{-1} B) =0.
\end{equation}
From the above expressions for $A$, $B$, $C$ and $D$ we see that
\[ D-\lambda I-C(A-\lambda I)^{-1}B = \frac{1}{z^3} D_0 -\lambda I + O(1/|z|^4), \]
and thus, calculating the determinant on the right,
\begin{equation} \label{eq:calcdet}
0 = \lambda^2 - \frac{1}{z^6}(a-b)^2 + O(1/|z|^7) + O(\lambda/|z|^4).
\end{equation}
Using that $\lambda = O(1/|z|)$ in the last term gives $\lambda^2=O(1/|z|^5)$ and thus the improved bound $\lambda=O(1/|z|^{5/2})$. Therefore (\ref{eq:calcdet}) becomes
\begin{equation} \label{eq:calcdet2}
0 = \lambda^2 -\frac{1}{z^6} (a-b)^2 + O(1/|z|^{13/2}).
\end{equation}
Applying Rouch\'e's Theorem to the function $f(\lambda)=\lambda^2 -\frac{1}{z^6}(a-b)^2$ and the contours $\gamma_{\pm}$ given by circles centered at $\pm \frac{a-b}{z^3}$ and radius $1/|z|^{3+\varepsilon}$, $0<\varepsilon < 1/2$, shows that (\ref{eq:calcdet2}) has one root each in the interior of the disjoint contours $\gamma_{\pm}$ for $|z|$ sufficiently large. This completes the proof of Proposition~\ref{pro:ab}.
\end{proof}

We now return to the full random operator $H_{\omega}=h_0 + \sum_{n\in \Z^2} \omega_n f_n$ from case (iii) of Proposition~\ref{pro:BC2}. Fix values of $a$ and $b$ with $a\not= b$ and let
\[ h_{\omega,L}^0 := h_{a,b} + V_{\omega,L},\]
where
\[ V_{\omega,L} := \sum_{|n|_{\infty}>L} \omega_n f_n.\]
By the resolvent identity we have
\begin{equation} \label{eq:residentity}
\chi_0 (h_{\omega,L}^0-z)^{-1} \chi_0 = \chi_0 (h_{a,b}-z)^{-1} \chi_0 - \chi_0 (h_{\omega,L}^0-z)^{-1} V_{\omega,L} (h_{a,b}-z)^{-1} \chi_0.
\end{equation}
Using Proposition~\ref{pro:ab}, fix $z$ with $|z|$ sufficiently large such that $\chi_0(h_{a,b}-z)^{-1}\chi_0$ is simple and let
\[ \delta := \min\{ |\lambda-\mu|: \,\lambda, \mu \:\mbox{eigenvalues of $\chi_0(h_{a,b}-z)^{-1}\chi_0$, $\lambda \not= \mu$}\}.\]

We have supp$\,V_{\omega,L} \subset \Z^2 \setminus \left( \bigcup_{|n|_{\infty} \le L} C_n \right)$. Thus we find from a Combes-Thomas type estimate (see e.g.\ Chapter~11 of \cite{Kirsch} for a proof in the setting of discrete Schr\"odinger operators) that there are $C<\infty$ and $\eta>0$ such that
\begin{eqnarray} \label{eq:combesthomas}
\lefteqn{\|\chi_0 (h_{\omega,L}^0-z)^{-1} V_{\omega,L} (h_{a,b}-z)^{-1} \chi_0\|} \\ & \le & \|\chi_0 (h_{\omega,L}^0-z)^{-1} \sqrt{V_{\omega,L}}\| \|\sqrt{V_{\omega,L}} (h_{a,b}-z)^{-1} \chi_0\| \nonumber \\
& \le & (Ce^{-\eta L})^2 \nonumber
\end{eqnarray}
for all $L\in \N$ and uniformly in all $(\omega_n)_{|n|_{\infty}>L}$ with $\omega_n \in \mbox{supp}\,\rho$. Now fix $L$ sufficiently large such that the right hand side of (\ref{eq:combesthomas}) is less than $\delta/2$. By (\ref{eq:residentity}) we conclude that $\chi_0 (h_{\omega,L}-z)^{-1} \chi_0$ is simple.

To complete the proof we now use analyticity of $(H_{\omega}-z)^{-1}$ in the finitely many variables $\omega_n$, $|n|_{\infty} \le L$:

Fix $(\omega_n)_{|n|_{\infty}>L}$ with $\omega_n \in \,\mbox{supp}\,\rho$. Let $S$ be the Sylvester matrix (\ref{sylvestermatrix}) of $\chi_0(H_{\omega}-z)^{-1} \chi_0$. Then $\det S$ is analytic in each of the variables $\omega_n$, $|n|_{\infty} \le L$. For the particular choice of these variables given by the potential $V_{a,b}$ we get from simplicity of $\chi_0 (h_{\omega,L}-z)^{-1} \chi_0$ that $S$ is non-zero. Using analyticity in each of the variables $\omega_n$, $|n|_{\infty}\le L$, iteratively we conclude that $\det S$ is non-zero for Lebesgue-a.e.\ $(\omega_n)_{|n|_{\infty} \le L} \in \R^{(2L+1)^2}$. As $\mu$ is absolutely continuous, this also holds with respect to the product measure on $\R^{(2L+1)^2}$ generated by $\mu$. As discussed in the proof of Lemma~\ref{lem:green}(c), a matrix is simple if and only if the determinant of its Sylvester matrix is non-zero. Recalling that the choice of $\omega_n\in \,\mbox{supp}\,\rho$, $|n|_{\infty}> L$, was arbitrary, this completes the proof of almost sure simplicity of $\chi_0(H_{\omega}-z)^{-1} \chi_0$ and therefore Proposition~\ref{pro:BC2} for case (iii).

\begin{appendix}

\section{Background}

For the sake of completeness, we use this appendix to collect some classical facts on boundary values of analytic functions on the upper half plane and derive the properties of boundary values of Green's function which were used above.

An analytic function $f:\C^+\rightarrow \C^+$ is called a Herglotz function. A function $f: \C^+ \rightarrow \C$ is said to be of bounded characteristic if there exist functions $g$ and $h$, both bounded
and analytic in $\mathbb{C}^+$, with
\begin{equation}
f(z)=\frac{g(z)}{h(z)} \quad \quad \mbox{for all $z \in \C^+$}.
\end{equation}

\begin{lemma}\label{lem:bc}
(a) The set of functions of bounded
  characteristic is closed under scalar multiplication, addition,
  and multiplication.

(b) Herglotz functions are of bounded characteristic.

(c) If $f$ has bounded characteristic, then
  $f(t+i0):=\lim_{\epsilon \downarrow 0}f(t+i\epsilon)$ exists for Lebesgue-almost every $t \in \R$.
  If $|\{t:f(t+i0)=0\}|>0$, then $f$ is identically zero.
\end{lemma}

\begin{proof}
Part (a) is elementary. To show (b), let $f$ be Herglotz and
\begin{equation}
g=\frac{f-i}{f+i}
\end{equation}
Then $g$ is bounded, analytic, and
\begin{equation}
f=\frac{-i(1+g)}{g-1}
\end{equation}
is the ratio of bounded analytic functions. Part (c) is a classical result which can be found, for example, in the books \cite{Privalov} or \cite{Garnett}.

\end{proof}

Below, we say that a square matrix is simple if all its generalized eigenspaces are one-dimensional (and thus eigenspaces).

\begin{lemma} \label{lem:green}
Let $k\in \N$ and $H$ be a $k\times k$-matrix-valued function on the upper half plane $\C^+$, such that all its entries are of bounded characteristic. Then
(a) $H(t+i0) := \lim_{\epsilon \downarrow 0} H(t+i\epsilon)$ exists in norm for Lebesgue-a.e.\ $t\in \R$,

(b) if $H(z)$ is invertible for at least one $z\in \C^+$, then $H(t+i0)$ is invertible for a.e.\ $t\in \R$,

(c) if $H(z)$ is simple for at least one $z\in \C^+$, then $H(t+i0)$ is simple for a.e.\ $t\in \R$.
\end{lemma}

\begin{proof}
(a) All matrix elements of $H(t+i\varepsilon)$ have boundary values for a.e.\ $t\in \R$ by Lemma~\ref{lem:bc}(c). This implies the existence of norm limits for the finite matrix $H$ and almost every $t$.

(b)  Let $d(z)=\det \, H(z)$.  Then $d(z)$ is a sum of products of matrix elements of $H(z)$ and thus of bounded characteristic. By assumption, $d(z)$ is not identically zero in $\C^+$. Therefore we conclude from Lemma~\ref{lem:bc} that $d(t+i0)= \det \,H(t+i0)$ exists and is non-zero for almost every $t \in \R$, proving the claim.

(c) We use the following general fact: Suppose that $C=(c_{ij})$ is a $k\times k$ matrix and $\lambda_1,
\ldots, \lambda_k$ are its eigenvalues counted with algebraic
multiplicity.  Let
\begin{equation}
  P_C(x) = \det\,(xI-C)=\sum_{n=0}^k a_n x^n
\end{equation}
be the corresponding characteristic polynomial with $a_k=1$ and set
\begin{equation}
  \mathcal{F}(C)=\prod_{i<j}(\lambda_j-\lambda_i)^2.
\end{equation}
Thus $C$ is simple if and only if $\mathcal{F}(C) \not= 0$. Moreover, $\mathcal{F}(C)=(-1)^{\frac{1}{2}k(k-1)} \det S(C)$, with the Sylvester matrix \footnote{http://en.wikipedia.org/wiki/Discriminant}
\begin{equation}\label{sylvestermatrix}
S(C)= \left(
\begin{array}{ccccccc}
a_k&a_{k-1}&\cdots&a_0&&&   \\
&   a_k   & a_{k-1}& \cdots & a_0 & &   \\
& & \ddots & \ddots & & \ddots &   \\
& & & a_k & a_{k-1} & \cdots & a_0   \\
ka_k & (k-1)a_{k-1} & \cdots & a_1 & & &   \\
&   ka_k   & (k-1)a_{k-1}& \cdots & a_1 & &   \\
& & \ddots & \ddots & & \ddots &   \\
& & & ka_k & (k-1)a_{k-1} & \cdots & a_1
\end{array} \right).
\end{equation}

Now we can argue similar to the proof of (b): If $C=H(z)$, then the coefficients $a_n(z)$ of the characteristic polynomial are polynomials in the matrix-elements of $H(z)$ and thus of bounded characteristic. Therefore $\mathcal{F}(H(z))$ is of bounded characteristic and, by assumption, not identically vanishing in $\C^+$. Its boundary value $\mathcal{F}(H(t+i0))$ is non-zero and thus $H(t+i0)$ simple for almost every $t\in \R$.
\end{proof}

We conclude by commenting on the existence of the boundary values (\ref{eq:BSboundval}) of the Birman-Schwinger operators for the three models considered in this paper. For Models A and B these operators are finite rank and thus the existence of boundary values is a special case of Lemma~\ref{lem:green}(a). For Model B the operators $G(z) = f^{1/2} (H_{\hat{\omega}}-z)^{-1} f^{1/2}$ are compact for $z\in \C^+$, which follows from standard relative compactness properties of Schr\"odinger operators, e.g.\ \cite{SimonSemi}. To see why boundary values exist we use the following well-known result, see e.g.\ \cite{Naboko}.

\begin{lemma} \label{lem:bval}
If $H(\cdot)$ is an analytic bounded operator-valued function in the upper half plane $\C^+ := \{z: \mbox{Im}\,z>0\}$ such that $H(z)$ is trace class and Im$\,H(z) \ge 0$ for all $z\in \C^+$, then $H(E+i0) := \lim_{\varepsilon \downarrow 0} H(E+i\varepsilon)$ exists in operator norm (in fact every Schatten class norm other than the trace class) for almost every $E\in \R$.
\end{lemma}

For Model C we have analyticity of $G(z)$ and Im$\,G(z)\ge 0$ in the upper half plane, but $G(z)$ is generally not trace class (other than for $d=1$). But we can argue as follows, inserting spectral projections:

For any finite interval $I=[a,b]$ let $I'=[a-1,b+1]$. For $E\in I$ consider
\begin{eqnarray*}
G(E+i\varepsilon) & = & f^{1/2} (H_{\hat{\omega}}-(E+i\varepsilon))^{-1} P_{I'}(H_{\hat{\omega}}) f^{1/2} \\ & & \mbox{} + f^{1/2} (H_{\hat{\omega}}-(E+i\varepsilon))^{-1} P_{\R \setminus I'}(H_{\hat{\omega}}) f^{1/2}.
\end{eqnarray*}
The second term trivially has a limit as $\varepsilon \downarrow 0$, while the first term falls into the class considered in Lemma~\ref{lem:bval}. One uses here that $f^{1/2} P_{I'}(H_{\hat{\omega}})$ is Hilbert-Schmidt, e.g.\ \cite{SimonSemi}. As a consequence, $G(E+i0)$ exists for almost every $E\in I$, and, by exhaustion, for almost every $E\in \R$.

\end{appendix}

\end{document}